\documentclass[submission]{eptcs}
\providecommand{\event}{AFL 2017} 
\usepackage{breakurl}             
\usepackage{underscore}           
\usepackage[utf8]{inputenc}

\usepackage {amsfonts}
\usepackage{leftidx}
\usepackage {amsthm}
\usepackage {cite}
\usepackage {amsmath}
\usepackage{mathtools}
\newtheorem{claim}{Claim}
\newtheorem{theorem}{Theorem}
\newtheorem{lemma}{Lemma}
\newtheorem{basicIdea}{Basic Idea}
\newtheorem{proofbi}{Proof (Basic Idea)}

\newtheorem{example}{Example}

\usepackage {amsmath}
\usepackage {amssymb}
\usepackage{mathtools}

\usepackage{xcolor}
\usepackage[ruled,noline]{algorithm2e}

\SetKwBlock{Begin}{}{}
    \SetKwComment{Comment}{/*   }{   */}
    	\SetKwInOut{Input}{Input}
    \SetKwInOut{Output}{Output}

\title{CD Grammar Systems with Two Propagating Scattered Context Components Characterize the Family of Context Sensitive Languages}
\author{Alexander Meduna
\institute{Department of Information Systems, Faculty of Information Technology,\\
Brno University of Technology\\
Božetěchova 2, Brno 612 66, Czech republic}
\email{meduna@fit.vutbr.cz}
\and
Jakub Marti\v{s}ko
\institute{Department of Information Systems, Faculty of Information Technology,\\
Brno University of Technology\\
Bo\v{z}et\v{e}chova 2, Brno 612 66, Czech republic}
\email{\quad imartisko@fit.vutbr.cz}
}
\def\titlerunning{Propagating Scattered Context Grammar Systems}
\def\authorrunning{A. Meduna \& J. Marti\v{s}ko}
\begin{document}
\maketitle

\newcommand{\marked}[5]{\prescript{#2}{#3}#1^{#4}_{#5}}

\newcommand{\mpone}{\triangle}
\newcommand{\mptwo}{\triangledown}
\newcommand{\mpthr}{\diamond}

\newcommand{\cf}[1]{\marked{#1}{}{|}{}{|}}
\newcommand{\cfn}[2]{\marked{#1}{}{|}{}{#2|}}

\newcommand{\csln}[2]{\marked{#1}{}{>}{}{#2|}}
\newcommand{\csrn}[2]{\marked{#1}{}{|}{}{#2<}}

\newcommand{\lcs}[1]{\marked{#1}{}{>}{}{|}}
\newcommand{\rcs}[1]{\marked{#1}{}{|}{}{<}}
\newcommand{\lcsc}[1]{\marked{#1}{}{>}{\wedge{}}{|}}
\newcommand{\rcsc}[1]{\marked{#1}{}{}{\wedge{}}{<}}
\newcommand{\cfc}[1]{\marked{#1}{}{}{\wedge}{|}}
\newcommand{\cfcn}[2]{\marked{#1}{}{}{\wedge}{#2|}}

\newcommand{\cur}[1]{\marked{#1}{}{}{\wedge}{}}
\newcommand{\pc}[1]{\marked{#1}{\mpone}{}{\wedge}{}}

\newcommand{\prcsc}[1]{\marked{#1}{\mpone}{}{\wedge}{<}}
\newcommand{\mrcsc}[1]{\marked{#1}{\mptwo}{}{\wedge}{<}}
\newcommand{\mrcs}[1]{\marked{#1}{\mptwo}{}{}{<}}
\newcommand{\mcfc}[1]{\marked{#1}{\mptwo}{}{\wedge}{<}}
\newcommand{\pcfc}[1]{\marked{#1}{\mpone}{}{\wedge}{|}}
\newcommand{\pcf}[1]{\marked{#1}{\mpone}{}{}{|}}
\newcommand{\mcf}[1]{\marked{#1}{\mptwo}{}{}{|}}

\newcommand{\fplus}[1]{\marked{#1}{\mpone}{}{}{}}
\newcommand{\fminus}[1]{\marked{#1}{\mptwo}{}{}{}}
\newcommand{\fstar}[1]{\marked{#1}{\mpthr}{}{}{}}

\newcommand{\trm}[1]{\marked{#1}{}{}{'}{}}
\newcommand{\cftrm}[1]{\marked{#1}{}{|}{'}{|}}
\newcommand{\mcftrm}[1]{\marked{#1}{\mptwo}{|}{'}{|}}
\newcommand{\mcftrmf}[1]{\marked{#1}{\mptwo}{}{'}{|}}
\newcommand{\strm}[1]{\marked{#1}{\mpthr}{}{'}{}}

\newcommand{\aorcs}[1]{\marked{#1}{1}{|}{}{<}}
\newcommand{\aolcs}[1]{\marked{#1}{1}{>}{}{|}}
\newcommand{\atrcs}[1]{\marked{#1}{2}{|}{}{<}}
\newcommand{\atlcs}[1]{\marked{#1}{2}{>}{}{|}}
\newcommand{\aux}[1]{\marked{#1}{\sim}{|}{}{|}}

\begin{abstract}
The $\mathcal{L}(PSCG)=\mathcal{L}(CS)$ problem asks whether propagating scattered context grammars and context sensitive grammars are equivalent. The presented paper reformulates and answers this problem in terms of CD grammar systems. More specifically, it characterizes the family of context sensitive languages by two-component CD grammar systems with propagating scattered context rules.
\end{abstract}

\section{Introduction}
Are propagating scattered context grammars as powerful as context sensitive grammars? This question customarily referred to as the $\mathcal{L}(PSCG)=\mathcal{L}(CS)$ (see \cite{greibach}) problem, represents a long standing open problem in formal language theory. The present paper reformulates and answers this question in terms of CD grammar systems.

More precisely, the paper introduces CD grammar systems whose components are propagating scattered context grammars. Then, it demonstrates that two-component grammar systems of this kind generate the family of context-sensitive languages, thus the answer to this problem is in affirmation if the problem is reformulated in the above way.

\section{Preliminaries}
We assume that the reader is familiar with formal language theory (see \cite{meduna2000,dassow1989regulated,handbookVol1,handbookVol2} for details). For an alphabet (finite nonempty set) $V$, $V^*$ represents the free monoid generated by $V$ under the operation of concatenation. The unit of $V^*$ is denoted by $\varepsilon$. The length of string $x_1\dots x_n \in V^*$ is denoted as $|x_1\dots x_n|$ and is equal to $n$. Similarly by $|x_1\dots x_n|_N$, the length of string when counting only symbols of $N$ is denoted. The function $alph(\alpha), \alpha \in V^*$ is defined as $alph(\alpha)=\{x:\alpha = \beta x \gamma; x \in V, \beta,\gamma \in V^*\}$.

A \emph{scattered context grammar (SCG)} is a quadruple $G=(N,T,P,S)$, where $N$ and $T$ are alphabets of nonterminal and terminal symbols respectively, where $N\cap T=\emptyset$, further let $V=N\cup T$. $S\in N$ is starting symbol. $P$ is a nonempty finite set of rules of the form $(A_{1},\dots,A_{n})\rightarrow(\alpha_{1},\dots\alpha_{n})$, where $A_{i}\in N$, $\alpha_{i}\in V^{*}, 1\leq i \leq n$, for some $n\geq 1$. Let $u,v \in V^{*}$, where $u=u_{1}A_{1}u_{2}A_{2}u_{3}\dots u_{n}A_{n}u_{n+1}$ and $v=u_{1}\alpha_{1}u_{2}\alpha_{2}u_{3}\dots u_{n}\alpha_{n}u_{n+1}$, $(A_{1},A_{2},A_{3},\dots,A_{n})\rightarrow (\alpha_{1},\alpha_{2},\alpha_{3},\dots,\alpha_{n})\in P$, where $u_i\in V^*$ for all $1\leq i\leq n+1$; then $u\Rightarrow v$ in $G$.

The language generated by SCG $G$ is defined as $L(G)=\{x:S\Rightarrow^{*}x, x\in T^{*}\}$, where
$\Rightarrow^{*}$ and $\Rightarrow^{+}$ denote the transitive-reflexive closure  and the transitive closure  of $\Rightarrow$, respectively. A SCG grammar is said to be \emph{propagating (PSCG)} iff each $(A_{1},\dots,A_{n})\rightarrow(\alpha_{1},\dots\alpha_{n})\in P$ satisfies $\alpha_{i}\neq \varepsilon, 1\leq i\leq n$. $\mathcal{L}(SCG)$ and $\mathcal{L}(PSCG)$ denote the families of languages generated by SCGs and PSCGs respectively.

A \emph{context-sensitive grammar (CSG)} is a quadruple $G=(N,T,P,S)$, where $N$ and $T$ are the alphabets of nonterminal and terminal symbols respectively, where $N\cap T=\emptyset$. Set $V=N\cup T$. $S\in N$ is the starting symbol. $P$ is a nonempty finite set of rules of the form $\alpha A\beta\rightarrow \alpha \gamma\beta$, where $A\in N$, $\alpha,\beta,\gamma \in V^{*}$ and $\gamma\neq\varepsilon$. Let $u,v \in V^{*}$, $u=u_{1}\alpha A\beta u_{2}$, $v=u_{1}\alpha \gamma\beta u_{2}$, $\alpha A\beta\rightarrow \alpha \gamma\beta \in P$	where $u_{1},u_{2},\alpha,\beta,\gamma\in V^{*}, A\in N$, $\gamma\neq\varepsilon$, then $u\Rightarrow v$ in $G$.

The language generated by CSG $G$ is defined as $L(G)=\{x:S\Rightarrow^{*}x, x\in T^{*}\}$, where
$\Rightarrow^{*}$ and $\Rightarrow^{+}$ denote the transitive-reflexive closure  and the transitive closure  of $\Rightarrow$ respectively. By $\mathcal{L}(CS)$ the family of languages generated by CSGs is denoted.

A grammar $G=(N,T,P,S)$ is in \emph{Kuroda Normal form} if every rule in $P$ has one of the following forms:
 \begin{enumerate}
  \item{$AB\rightarrow CD$}
  \item{$A\rightarrow CD$}
  \item{$A\rightarrow C$}
  \item{$A\rightarrow a$}
\end{enumerate}
where $A,B,C,D\in N$ and $a\in T$. Recall  that every CSG can be transformed into an equivalent grammar in Kuroda normal form (see \cite{meduna2000,KURODA1964207}). Without any loss of generality, for any CSG we assume their equivalent in the Kuroda normal form in what follows.

To emphasize that rule $p$ was used during the derivation step, we will sometimes write $\alpha \Rightarrow \beta [p]$.

A \emph{cooperating distributed grammar system (CDGS)} (see \cite{csuhaj1990cooperating,handbookVol1,csuhaj1994grammar}) of degree $n$ is $n$+3 tuple $G=(N,T,S,P_{1},P_{2},\dots,P_{n})$, where $N$ and $T$ are alphabets of nonterminal and terminal symbols respectively, where $N\cap T=\emptyset$, further let $V=N\cup T$. $S\in N$ is starting symbol. $P_{i}, 1\leq i \leq n$ are nonempty finite sets (called components) of rewriting rules over $V$. For a CDGS $G=(N,T,S,P_{1},P_{2},\dots,P_{n})$, the terminating ($t$) derivation by the $i$-th component, denoted as ${\Rightarrow^{t}_{P_{i}}}$ is defined as $u{\Rightarrow^{t}_{P_{i}}}v$ iff ${u\Rightarrow^{*}_{P_{i}}}v$ and there is no $z\in V^{*}$ such that $v\Rightarrow_{P_{i}} z$. The language generated by CDGS $G=(N,T,S,P_{1},P_{2},\dots,P_{n})$ working in $t$ mode is defined as
 $L(G)=\{x:S{\Rightarrow^{t}_{P_{i_{1}}}}x_{1}$ ${\Rightarrow^{t}_{P_{i_{2}}}}x_{2}$ $ \dots $ ${\Rightarrow^{t}_{P_{i_{m}}}}x, m\geq 1, 1\leq i_{j}\leq n, 1\leq j\leq m, x\in T^{*}\}$.

In this paper, CDGS with propagating scattered context rules (SCGS) and CDGS with context-sensitive rules will be considered.

\section{Main results}
In this section, the identity of $\mathcal{L}(SCGS)$ and $\mathcal{L}(CS)$ will be demonstrated.

\begin{lemma}
\label{lemma:scgssubsetcs}
$\mathcal{L}(SCGS) \subseteq \mathcal{L}(CS)$
\end{lemma}
\begin{proof}
Recall that \cite{dassow1989regulated} shows that any scattered context grammar can be simulated by context-sensitive grammar. Similarly, \cite{csuhaj1994grammar} shows that any CDGS with context-sensitive components working in $t$ mode can be transformed to equivalent CSG. Based on those two facts, it is easy to show that any SCGS can be simulated by CSG.
\end{proof}

\begin{lemma}
\label{lemma:cssubsetscgs}
$\mathcal{L}(CS) \subseteq \mathcal{L}(SCGS)$
\end{lemma}
Take any CSG $G=(N,T,P,S)$ satisfying Kuroda normal form. An equivalent SCGS $\Gamma=(N_{GS}, T,$ $ \fplus{S}, P_{1}, P_{2})$ can be constructed using the following constructions. Set $N_{GS} = N\cup \{!\} \cup N_{T}\cup N_{first} \cup N_{CF} \cup N_{CS} \cup N_{cur}$ ($!\notin N\cup T$). Where:

  {\small\begin{flalign*}
      N_{T}=&\{\trm{a}:\forall a \in T\}\\
      N_{\mpone}=&\{\fplus{X}:\forall X\in N\cup N_{T}\}\\
      N_{\mptwo}=&\{\fminus{X}:\forall X\in N\cup N_{T}\}\\
    N_{\mpthr}=&\{\fstar{X}:\forall X\in N\cup N_{T}\}\\
    N_{first}=&N_{\mpone}\cup N_{\mptwo}\cup N_{\mpthr}\\
    N_{CF\mpone}=&\{{X_{|}}:\forall X\in N_{\mpone}\}\\
    N_{CF\mptwo}=&\{{X_{|}}:\forall X\in N_{\mptwo}\}\\
    N_{CF\mpthr}=&\{{X_{|}}:\forall X\in N_{\mpthr}\}\\
    N_{CF}=&\{\cf{X}:\forall X\in N\} \cup N_{CF\mpone}\cup N_{CF\mptwo} \cup N_{CF\mpthr}\\
	  N_{CS}=&\{\rcs{X}:\forall X\in N\} \cup  \{{X_{<}}:\forall X\in N_{first}\} \cup\{\lcs{X}:\forall X\in N\}\\
    N_{cur}=&\{\rcsc{X}:\forall X\in N \cup N_{first}\}\cup\{\cfc{X}:\forall X\in N \cup N_{first}\}.
      \end{flalign*}}%

  Analogically to sets $N_{CF\mpone}, N_{CF\mptwo}$ and	$N_{CF\mpthr}$, we call subsets of $N_{CS}$ and $N_{cur}$ constructed using the set $N_{\mpone}$ as $N_{CS\mpone}$, and $N_{cur\mpone}$, respectively. We use similar naming convention for subsets constructed using the $N_{\mptwo}$ and $N_{\mpthr}$.

  Set $P_{1}$ to the union of the following sets:

  \begin{flalign*}
  P_{T}^{1}=&\{(\fplus{X})\rightarrow (\fstar{X}):\forall X\in N\cup N_{T}\}\\
   \cup& \{(\fstar{X},\trm{a})\rightarrow (\fstar{X},a):\\&\forall X\in N\cup N_{T},\forall \trm{a} \in N_{T}\}\\
   \cup& \{(\strm{a})\rightarrow (a):\forall \trm{a}\in N_{T}\}\\
   P_{AtoBC}^{1}=&\{(\fplus{X},A)\rightarrow (\mcf{X},\cf{B}\cf{C}): \\&\forall X\in N_{T}\cup N, \forall p\in P, p=A\rightarrow BC\}\\
   \cup& \{(\fplus{A})\rightarrow (\mcf{B}\cf{C}): \forall p\in P, p=A\rightarrow BC\}\\
   P_{AtoB}^{1}=&\{(\fplus{X},A)\rightarrow (\mcf{X},\cf{B}): \\&\forall X\in N_{T}\cup N, \forall p\in P, p=A\rightarrow B\}\\
    \cup&\{(\fplus{A})\rightarrow (\mcf{B}): \forall p\in P, p=A\rightarrow B\}\\
  \end{flalign*}
     \begin{flalign*}
  P_{Atoa}^{1}=&\{(\fplus{X},A)\rightarrow (\mcf{X},\cftrm{a}): \\&\forall X\in N_{T}\cup N, \forall p\in P, p=A\rightarrow a\}\\
  \cup&\{(\fplus{A})\rightarrow (\mcftrmf{a}): \forall p\in P, p=A\rightarrow a\}\\
  P_{ABtoCD}^{1}=&\{(\fplus{X},A,B)\rightarrow (\mcf{X},\rcs{C},\lcs{D}): \\&\forall X\in N_{T}\cup N, \forall p\in P, p=AB\rightarrow CD\}\\
  \cup&\{(\fplus{A},B)\rightarrow (\mrcs{C},\lcs{D}): \\&\forall p\in P, p=AB\rightarrow CD\}\\
  P_{phase2}^{1}=&\{(X,B)\rightarrow (X,\cf{B}): \\&\forall B \in N \cup N_{T},X \in N_{CS\mptwo}\cup N_{CF\mptwo}\}
  \end{flalign*}

set $P_{2}$ to the union of these subsets:

\begin{flalign*}
  P_{init}^{2}=&\{(\mcf{X})\rightarrow (\pcfc{X}): \forall X\in N_{T}\cup N\}\\
  \cup& \{(\mrcs{X})\rightarrow (\prcsc{X}): \forall X\in N_{T}\cup N\}\\
  P_{check}^{2}=&\{(\cfc{A},\cf{B})\rightarrow (A,\cfc{B}): \forall X,A,B\in N_{T}\cup N\}\\
  \cup& \{(\rcsc{A},\lcs{B})\rightarrow (A,\cfc{B}): \forall X,A,B\in N_{T}\cup N\}\\
  P_{checkf}^{2}=&\{(\pcfc{A},\cf{B})\rightarrow (\fplus{A},\cfc{B}): \forall A,B\in N_{T}\cup N\}\\
  \cup& \{(\prcsc{A},\lcs{B})\rightarrow (\fplus{A},\cfc{B}): \forall A,B\in N_{T}\cup N\}\\
  P_{end}^{2}=&\{(\pc{A})\rightarrow (\fplus{A}): \forall A\in N_{T}\cup N\}\\
  \cup& \{(\fplus{A},\cfc{B})\rightarrow (\fplus{A},B): \forall A,B\in N_{T}\cup N\}\\
  \cup& \{(\pcfc{A})\rightarrow (\fplus{A}): \forall A\in N_{T}\cup N\}\\
  P_{block}^{2}=&\{(\cf{X})\rightarrow (!): \forall X \in N_{T}\cup N\}\\
  \end{flalign*}

\begin{basicIdea}We will now briefly describe how the resulting SCGS $\Gamma$ simulates the input CSG $G$. The system consists of two components, both working in $t$ mode. The computation of $\Gamma$ consists of two phases. During the first one, all terminals are represented by a nonterminal variant of themselves. The simulation itself takes place during the first phase. 

 The simulation in $\Gamma$ of each application of one rule of $G$ consists of two parts. Firstly, the first component applies the selected rule using the modified nonterminals contained in the sets $N_{CS}$ and $N_{CF}$. Symbols of the type $\rcs{X}$ denote that the rewriting is done in a context-sensitive way and that the remaining symbol on the right hand side of the rule should appear immediately right of the symbol. Similarly $\lcs{X}$ denotes that the rest of the right hand side of the rule should appear immediately left of the symbol. Symbols of the form $\cf{X}$ then represent context-free rewriting. After the application of the rule, the first component rewrites all remaining symbols to their context-free variant and then deactivates. This is done using the rules of the set $P_{phase2}^{1}$.  The fact that only one of the rules was applied is checked using the first symbol of the sentential form. This symbol is of the form $\fplus{X}$ or $\fminus{X}$ (plus the context-sensitive and context-free versions), where the marks $\mpone$ and $\mptwo$ indicate, whether next rule should be simulated, or remaining symbols should be rewritten to their context-free variant.

 The second component then checks, whether the first component applied the rule correctly. This is done using the special $\wedge$ mark. This symbol indicates, which symbol is currently checked, we will call this symbol \emph{current symbol}. Symbols are checked in pairs, where the first symbol of the pair is the current symbol and the second symbol is some symbol right of the first one. During this check, the special marks ($|,<,>$) on the adjacent sides of those symbols are checked and removed and the $\wedge$ mark is moved to the other symbol of the pair. Since the first symbol of the pair is always the current symbol, the $\wedge$ moves from the left side of the sentential form to the right, with no way of returning back left. When all of the symbols are checked, the second component is deactivated and the first one simulates new rule. Since the components have scattered context rules, it is not guaranteed that adjacent symbols are always checked by the second component. Because of this, set of rules $P_{block}^{2}$ is created. When some of the symbols is skipped during the checking phase, these rules will block the generation of sentence by $\Gamma$.

 The second phase, which rewrites all nonterminals to terminals, is started by rewriting of the first symbol $\fplus{X}$ to $\fstar{X}$. Then for each symbol $\trm{a}$, there is a rule of the form $(\fstar{X},\trm{a})\rightarrow (\fstar{X},a)$, where $a$ is corresponding terminal symbol. Finally, the leftmost symbol itself is rewritten to its terminal form. Since all the rules of all components always check the first symbol, after this step no further rewriting can be done and all nonterminals that remain in the sentential form cannot be removed. This phase is represented by set $P_{T}^{1}$.
\end{basicIdea}

Next, we sketch a formal proof that $L(G)=L(\Gamma)$. Its fully rigorous version is left to the reader.

 \begin{claim}
 \label{claimmarks}
In any sentential form, there is always at most one symbol marked with any of $\mpone, \mptwo, \mpthr$.
\end{claim}

\begin{proof}
Observe that no rule contains more than one symbol marked with any of $\mpthr,\mptwo,\mpone$ on the right hand side. Furthermore observe that if any marked symbol does appear on the right hand side of a rule, there is also a marked symbol on the left hand side of the same rule. Thus no new marked symbols can be introduced into the sentential form.
\end{proof}

 \begin{claim}
Any derivation that generates a sentence ends with a sequence of rules of the form $p_1p_{2_{1}}\dots p_{2_{n}}p_3$, where $p_1,p_{2_{i}},p_3\in P_{T}^{1}, 1\leq i\leq n, n\geq0$, where $p_1$, $p_{2_{i}}$ and $p_3$ are from the first, second and third subset of $P_{T}^{1}$, respectively.  No rule from $P_{T}^{1}$ is applied before this sequence.
\end{claim}

\begin{proof}
  Recall that only rules which have terminals on the right hand side are in set $P_{T}^{1}$ which is defined as follows (in this proof, we named each of its subsets for the sake of simplicity):
  \begin{flalign*}
  P_{T}^{1}=&{}_1P_{T}^{1}\cup{}_2P_{T}^{1}\cup{}_3P_{T}^{1}\\
  {}_1P_{T}^{1}=&\{(\fplus{X})\rightarrow (\fstar{X}):\forall X\in N\cup N_{T}\}\\
  {}_2P_{T}^{1}=&\{(\fstar{X},\trm{a})\rightarrow (\fstar{X},a):\forall X\in N\cup N_{T},\forall \trm{a} \in N_{T}\}\\
  {}_3P_{T}^{1}=&\{(\strm{a})\rightarrow (a):\forall \trm{a}\in N_{T}\}
  \end{flalign*}

  Suppose any sentential form $\chi$ such that $\chi=\fplus{x}{}_0^{'}\trm{x}{}_1\dots\trm{x}{}_n$, where $\trm{x}{}_i\in N_T$ $0\leq i \leq n$ and $\fplus{x}{}_0^{'} \in N_{\mpthr}$. Observe that all rules that do rewriting to terminals check the existence of a symbol $\fstar{X}$ in the sentential form. This symbol is created in a following way: $$\fplus{x}{}_0^{'}\trm{x}{}_1\dots\trm{x}{}_n\Rightarrow \fstar{x}{}_0^{'}\trm{x}{}_1\dots\trm{x}{}_n [p], p\in {}_1P_{T}^{1}$$

  Careful examination of sets $P_1$ and $P_2$ shows that only rules with $\fstar{x}{}_0^{'}$ on its left hand side are in sets ${}_2P_{T}^{1}$ and ${}_3P_{T}^{1}$. Suppose the following derivation $$\fstar{x}{}_0^{'}\trm{x}{}_1\dots\trm{x}{}_n [p], p\in {}_1P_{T}^{1}\Rightarrow x{}_0\trm{x}{}_1\dots\trm{x}{}_n [p], p \in {}_3P_{T}^{1}, x_0 \in T$$
  Based on the claim \ref{claimmarks}, $alph(\trm{x}{}_1\dots\trm{x}{}_n)\cap (N_{\mpone}\cup N_{\mptwo} \cup N_{\mpthr})=\emptyset$. Each rule $p\in P_1$ contains some symbol from $N_{\mpone}\cup N_{\mptwo} \cup N_{\mpthr}$ on its left hand side. There is thus no $\chi^{'}\neq\chi, \chi = x{}_0\trm{x}{}_1\dots\trm{x}{}_n $ such that $\chi\Rightarrow\chi^{'}[p], p\in P_1$. Further no rule from $P_2$ can be used (see claim \ref{claim:secondComp}). For any successful derivation $p\in {}_3P_{T}^{1}$ must thus be used as a last rule of this derivation.

  Suppose $\fplus{x}{}_0^{'}\trm{x}{}_1\dots\trm{x}{}_n\Rightarrow \fstar{x}{}_0^{'}\trm{x}{}_1\dots\trm{x}{}_n [p], p\in {}_1P_{T}^{1}$, and further let  $\chi_1=\fstar{x}{}_0^{'}\trm{x}{}_1\dots\trm{x}{}_n$ where $|\chi_1|>1$. Based on the previous paragraph, in any successful derivation, the following sequence of rules has to be applied $$\chi_1\Rightarrow \chi_2 [p_1] \Rightarrow \dots \Rightarrow \chi_n [p_n]$$
  where $\chi_n=\fstar{x}{}_0^{'}{x}{}_1\dots{x}{}_n [p_n]$, $p_i \in {}_2P_{T}^{1}$. For each $\chi_i$ and $\chi_{i+1}$ following holds $|\chi_i|_T = |\chi_{i+1}|_T -1$.

\end{proof}

We have just shown that the rules from $P_{T}^{1}$ are only applied right before the end of the successful simulation. Consequently, we do not mention this subset in any of the following proofs.

\begin{claim}
\label{firstcompsinglerule}
The first component of $\Gamma$ rewrites sentential forms of the form $\fplus{X}\alpha$ to a string of one of the following forms
\begin{enumerate}
\item{$\mcf{Y}\beta$}
\item{$\mrcs{Y}\gamma$}
\end{enumerate}
where $X,Y\in N \cup N_T$, $\alpha \in (N\cup N_{T})^*$, $\beta,\gamma \in (N_{CS}\cup N_{CF})^{*}$ (such that Claim \ref{claimmarks} holds) where either (a) or (b) given next is true:
\begin{enumerate}
\item[(a)]{$\beta\in (N_{CF})^{*}$}
\item[(b)]{$\beta= Y_0\dots \rcs{U}\dots\lcs{V}\dots Y_n$, where $Y_i \in N_{CF}, 0\leq i\leq n$ and $\rcs{U},\lcs{V}\in N_{CS}$}
\end{enumerate}
and $\gamma = Y_0\dots\lcs{V}\dots Y_n$, where $Y_i \in N_{CF}, 0\leq i\leq n$ and $\lcs{V}\in N_{CS}$.
\end{claim}
\begin{proof}
Consider sentential form $\fplus{X}\alpha$ defined as above, where $\alpha=X_0\dots X_n$. Since there is no symbol from the alphabet $N_{cur}\cup N_{\mptwo}$ only rules of the first component can be used.

From $\fplus{X}\alpha$, $\Gamma$ makes a derivation step in one of the following eight ways (each derivation corresponds to one subset of the rules of the first component of $\Gamma$):
\begin{enumerate}
\item{\label{enum:1:1}$\fplus{X}X_0\dots X_{i-1}AX_{i+1}\dots X_n\\\Rightarrow \mcf{X}X_0\dots X_{i-1}\cf{B}\cf{C}X_{i+1}\dots X_n$}
\item{\label{enum:1:2}$\fplus{A}X_0\dots X_n\\\Rightarrow \mcf{B}\cf{C}X_0\dots X_n$}
\item{\label{enum:1:3}$\fplus{X}X_0\dots X_{i-1}AX_{i+1}\dots X_n\\\Rightarrow \mcf{X}X_0\dots X_{i-1}\cf{B}X_{i+1}\dots X_n$}
\item{\label{enum:1:4}$\fplus{A}X_0\dots X_n\\\Rightarrow \mcf{B}X_0\dots X_n$}
\item{\label{enum:1:5}$\fplus{X}X_0\dots X_{i-1}AX_{i+1}\dots X_n\\\Rightarrow \mcf{X}X_0\dots X_{i-1}\cftrm{a}X_{i+1}\dots X_n$}
\item{\label{enum:1:6}$\fplus{A}X_0\dots X_n\\\Rightarrow \mcftrmf{a}X_0\dots X_n$}

\item{\label{enum:1:7}$\fplus{X}X_0\dots X_{i-1}A\delta BX_{i+1}\dots X_n\\\Rightarrow \mcf{X}X_0\dots X_{i-1}\rcs{C}\delta \lcs{D}X_{i+1}\dots X_n$}
\item{\label{enum:1:8}$\fplus{A}\delta BX_0\dots X_n\\\Rightarrow \mrcs{C}\delta \lcs{D}X_0\dots X_n$}
\end{enumerate}
where $\delta\in (N\cup N_T)^*$. Observe that each of the generated strings is in one of the following forms:
\begin{itemize}
\item{$\mcf{X}\delta_{1}B\delta_{2}C\delta_{3}$ (\ref{enum:1:1}-\ref{enum:1:7})}
\item{$\mrcs{X}\delta_{1}\lcs{D}\delta_{2}$ (\ref{enum:1:8})}
\end{itemize}
where $\delta_{1},\delta_{2},\delta_{3}\in (N\cup N_T)^*$, $D\in N_{CS}$ (the second subset) and either $B,C \in N_{CF}$ or $B\in N_{CS}$ (the first subset) and $C\in N_{CS}$ (the second subset).

After this first rule is applied, the sentential form contains symbol marked with $\mptwo$. Since the components of $\Gamma$ work in $t$ mode, rules of the first component have to be applied as long as there are some symbols that can be rewritten. This means that the rules from the set $P_{phase2}^{1}$ have to be used now. Because $\delta_{1},\delta_{2},\delta_{3}\in (N\cup N_T)^*$ and the left hand sides of the rules from $P_{phase2}^{1}$ are defined for all symbols in $N\cup N_T$. Substring $\delta_{1}=Z_0\dots Z_n$, $Z_i \in N\cup N_T$ is rewritten to $\delta_{1}^{'}={}_{|}Z_{0|}\dots {}_{|}Z_{n|}$, $Z_i \in N\cup N_T, 0\leq i\leq n$. The same applies to $\delta_{2},\delta_{3}$. By using  $P_{phase2}^{1}$ we obtain one of the following sentential forms:
\begin{itemize}
\item[]{$\mcf{X}\delta_{1}B\delta_{2}C\delta_{3}\Rightarrow^{*}\mcf{Y}\beta$ }
\item[]{$\mrcs{X}\delta_{1}\lcs{D}\delta_{2}\Rightarrow^{*}\mrcs{Y}\gamma$ }
\end{itemize}
\end{proof}

\begin{claim}
During its activation, the first component applies no more than one rule of the simulated CSG. This follows from Claim \ref{firstcompsinglerule} and its proof.
\end{claim}

\begin{claim}
\label{claim:secondComp}
The second component of $\Gamma$ rewrites any sentential form of the form $\mcf{X}\cfn{X}{0}\dots \cfn{X}{n}$ to a string of the form $\fplus{X}X_1\dots X_n$, where $X_i\in N\cup N_T, 0\leq i \leq n$ .
\end{claim}

\begin{proof}
Suppose sentential form\footnote{The case where $|\chi|$=1 is trivial and is left to the reader.} $\chi=\mcf{X}\cfn{X}{0}\dots \cfn{X}{n}$ where $X_i\in N\cup N_T, 0\leq i \leq n$. Observe that $|\chi|_{N_{cur}}=0$. Only rules\footnote{We ignore the set of blocking rules $P_{block}^2$ for now.} that can be used are thus from the first subset of $P_{init}^2$. This leads to $$\chi=\mcf{X}\cfn{X}{0}\dots \cfn{X}{n}\Rightarrow \chi_0=\pcfc{X}\cfn{X}{0}\dots \cfn{X}{n}$$

The only rule applicable to $\chi_0$ must be from the set $P_{checkf}^{2}$. This leads to:
 $$\chi_0=\pcfc{X}\cfn{X}{0}\dots \cfn{X}{n}\Rightarrow \chi_1=\fplus{X}\alpha_1 \cfcn{X}{i_1} \alpha_2$$
where $\alpha_1, \alpha_2 \in N_{CF}^{*}$. Again, careful observation of rules of the set $P_2$ shows that only rules from the set $P_{check}^{2}$ and $P_{end}^{2}$ may be used. The first option leads to following derivations:
$$\fplus{X}\alpha_1^{1} \cfcn{X}{i_1} \alpha_1^{2}\Rightarrow \fplus{X}\alpha_1^1 X_{i_1} \alpha_2^1 \cfcn{X}{i_2} \alpha_2^2\Rightarrow\dots\Rightarrow\fplus{X}\alpha_1^1X_{i_1}\alpha_2^1 X_{i_2} \dots \alpha_n^1\cfcn{X}{i_n}\alpha_{n}^2$$
where $\alpha_{k}^{j}\in N_{CF}^*, 1\leq k \leq n, j\in \{1, 2\}$. Further, there is no rule $p$ such that:
$$\fplus{X}\alpha_1^1X_{i_1}\alpha_2^1 X_{i_2} \dots \alpha_n^1\cfcn{X}{i_n}\alpha_{n}^2\Rightarrow \fplus{X}\alpha_1^1X_{i_1}\alpha_2^1 X_{i_2} \dots \cfc{Y}\dots \alpha_n^1{X}_{i_n}\alpha_{n}^2 [p]$$

Suppose $\fplus{X}\alpha_1^1X_{i_1}\alpha_2^1 X_{i_2} \dots \alpha_n^1\cfcn{X}{i_n}\alpha_{n}^2$ and rule $p\in P_{end}^2$:
$$\fplus{X}\alpha_1^1X_{i_1}\alpha_2^1 X_{i_2} \dots \alpha_n^1\cfcn{X}{i_n}\alpha_{n}^2\Rightarrow \fplus{X}\alpha_1^1X_{i_1}\alpha_2^1 X_{i_2} \dots \alpha_n^1{X}{i_n}\alpha_{n}^2 [p]$$

Suppose that adjacent symbols were always rewritten during the application of rules from the sets $P_{checkf}^2$ and $P_{check}^2$. This would mean that $\alpha_k^{j}=\varepsilon,  1\leq k\leq n, j\in \{1,2\}$ and we would thus obtain the desire sentential form $\fplus{X}X_1\dots X_n$, where $X_i\in N\cup N_T, 0\leq i \leq n$.

If, on the other hand, there was some $\alpha_l^{m}\neq \varepsilon, 1\leq l \leq n, m\in{1,2}$ this would mean that $$\pcfc{X}\cfn{X}{0}\dots \cfn{X}{n}\Rightarrow^* \fplus{X}\alpha$$ where $|\alpha|_{N_{cur}}=0$ and $|\alpha|_{N_{CF}}>0$. Since both GS components work in the $t$-mode, and there is some symbol from $N_{CF}$, the blocking symbols have to be introduced by the rules of the $P_{block}^2$ set. Because $|\alpha|_{N_{cur}}=0$, no other rules can be used on this form.
\end{proof}

\begin{claim}
\label{secondCompCS1}
The second component of $\Gamma$ rewrites any string of the form $\mrcs{X}\cfn{X}{0}$ $\dots\cfn{X}{j-1} \csln{X}{j}\dots \cfn{X}{n}$ to a string of the form $\fplus{X}X_1\dots X_n$, where $X_i\in N\cup N_T$, for all $i: 0\leq i \leq n$  if and only if $\cfn{X}{0}\dots\cfn{X}{j-1} =\varepsilon$; otherwise, blocking symbols are introduced.
\end{claim}

\begin{proof}
Proof of Claim \ref{secondCompCS1}  is similar to proof of Claim \ref{claim:secondComp} and is left to the reader.
\end{proof}

\begin{claim}
\label{secondCompCS2}
The second component of $\Gamma$ rewrites any string of the form $\mcf{X}\cfn{X}{0}$ $\dots \csrn{X}{j}\cfn{X}{j+1}  \dots$ $\cfn{X}{k-1}\csln{X}{k}$ $\dots$ $\cfn{X}{n}$ to a string of the form $\fplus{X}X_1\dots X_n$, where $X_i\in N\cup N_T$, for all $i: 0\leq i \leq n$ if and only if $\cfn{X}{j+1}\dots\cfn{X}{k-1} =\varepsilon$; otherwise blocking symbols are introduced.
\end{claim}

\begin{proof}
Proof of Claim \ref{secondCompCS2} is similar to proof of Claim \ref{claim:secondComp} and is left to the reader.
\end{proof}

Based on the previous claims, it is easy to show that each simulation of a rule of $G$ consists of a single activation of the first component followed by a single activation of the second component of $\Gamma$. If the simulated context-sensitive rule is applied in a scattered way, blocking symbols are introduced to the sentential form; otherwise the sentential form is prepared for the simulation of another rule. In the end, all nonblocking symbols are rewritten to terminals thus producing a sentence of the simulated language. Therefore, ${L}(G)={L}(\Gamma)$.

\begin{example}
Suppose CSG $G=(\{A,B,C,D,E\}, \{b, c,$ $d,$ $e\}, P, A)$ with rules $P=\{A\rightarrow BC, C\rightarrow CD, BD \rightarrow DB, CD \rightarrow ED, B\rightarrow b, C\rightarrow c, D\rightarrow d, E\rightarrow e\}$. Observe that there is no sentential form that could be generated by grammar $G$ where the rule $BD \rightarrow DB$ could be applied.

Based on the described constructions, equivalent SCGS $\Gamma$ can be created as $\Gamma = (N_{GS}, T, \fplus{A}, P_1, P_2)$.  We will now try to show, how would $\Gamma$ simulate $G$. Because the amount of rules and symbols created by the transformation algorithm is quite large, we will not list elements of these sets.

The only rule of $G$ that has starting symbol on its left hand side is $A\rightarrow BC$. Similarly, only rule applicable on $BC$ (we will ignore rules with terminals) is rule $C\rightarrow CD$ Derivation $A\Rightarrow^* BCD$ would be simulated using following sequence of derivation steps:

{\footnotesize\begin{flalign*}
\fplus{A}&\Rightarrow\mcf{B}\cf{C}&[(\fplus{A})\rightarrow (\mcf{B}\cf{C})\in P_{AtoBC}^{1}]\\
&\Rightarrow\pcfc{B}\cf{C}&[(\mcf{B})\rightarrow (\pcfc{B})\in P_{init}^{2}]\\
&\Rightarrow\fplus{B}\cfc{C}&[(\pcfc{B},\cf{C})\rightarrow (\fplus{B},\cfc{C})\in P_{checkf}^{2}]\\
&\Rightarrow\fplus{B}C&[(\fplus{B},\cfc{C})\rightarrow (\fplus{B},C)\in P_{end}^{2}]\\
\end{flalign*}}
This way, the first rule is simulated. It is important to note that since $\Gamma$ works in $t$ mode, rules from set $P_2$ are all applied together. The derivation would continue using the following rules:

{\footnotesize\begin{flalign*}
\fplus{B}C&\Rightarrow\mcf{B}\cf{C}\cf{D}&[(\fplus{B},C)\rightarrow (\mcf{B},\cf{C}\cf{D})\in P_{AtoBC}^{1}]\\
&\Rightarrow\pcfc{B}\cf{C}\cf{D}&[(\mcf{B})\rightarrow (\pcfc{B})\in P_{init}^{2}]\\
&\Rightarrow\fplus{B}\cfc{C}\cf{D}&[(\pcfc{B},\cf{C})\rightarrow (\fplus{B},\cfc{C})\in P_{checkf}^{2}]\\
&\Rightarrow\fplus{B}C\cfc{D}&[(\cfc{C},\cf{D})\rightarrow (C,\cfc{D})\in P_{check}^{2}]\\
&\Rightarrow\fplus{B}CD&[(\fplus{B},\cfc{D})\rightarrow (\fplus{B},D)\in P_{end}^{2}]\\
\end{flalign*}}%
As was mentioned before, rule $BD \rightarrow DB$ can in fact never be applied by the grammar $G$. Suppose sentential form $\fplus{B}CD$ of the $\Gamma$. Simulation of this rule would lead to the following derivation:

{\footnotesize\begin{flalign*}
\fplus{B}CD&\Rightarrow\mrcs{D}C\lcs{B}&[(\fplus{B},D)\rightarrow (\mrcs{D},\lcs{B})\in P_{ABtoCD}^{1}]\\
&\Rightarrow\mrcs{D}\cf{C}\lcs{B}&[(\mrcs{D},C)\rightarrow (\mrcs{D},\cf{C})\in P_{phase2}^{1}]\\
&\Rightarrow\prcsc{D}\cf{C}\lcs{B}&[(\mrcs{D})\rightarrow (\prcsc{D})\in P_{init}^{2}]\\
&\Rightarrow\fplus{D}\cf{C}\cfc{B}&[(\prcsc{D},\lcs{B})\rightarrow (\fplus{D},\cfc{B})\in P_{checkf}^{2}]\\
&\Rightarrow\fplus{B}\cf{C}D&[(\fplus{B},\cfc{D})\rightarrow (\fplus{B},D)\in P_{end}^{2}]\\
&\Rightarrow\fplus{B}!D&[(B)\rightarrow (!)\in P_{block}^{2}]
\end{flalign*}}%
Again, each component of $\Gamma$ works in $t$ mode. This ensures that any symbols skipped during the checking phase, will be replaced by blocking symbols (!) before the second component of $\Gamma$ deactivates.

On the other hand, rule $CD \rightarrow ED$ can be applied. The simulation of this rule works as follows:

{\footnotesize\begin{flalign*}
\fplus{B}CD&\Rightarrow\mcf{B}\rcs{E}\lcs{D}&[(\fplus{B},C,D)\rightarrow (\mcf{B},\rcs{E},\lcs{D})\in P_{ABtoCD}^{1}]\\
&\Rightarrow\pcfc{B}\rcs{E}\lcs{D}&[(\mcf{B})\rightarrow (\pcfc{B})\in P_{init}^{2}]\\
&\Rightarrow\fplus{B}\rcsc{E}\lcs{D}&[(\pcfc{B},\rcs{E})\rightarrow (\fplus{B},\rcsc{E})\in P_{checkf}^{2}]\\
&\Rightarrow\fplus{B}E\cfc{D}&[(\rcsc{E},\lcs{D})\rightarrow ({E},\cfc{D})\in P_{check}^{2}]\\
&\Rightarrow\fplus{B}{E}D&[(\fplus{B},\cfc{D})\rightarrow (\fplus{B},D)\in P_{end}^{2}]\\
\end{flalign*}}%
\end{example}

\begin{theorem}
$\mathcal{L}(SCGS)=\mathcal{L}(CS)$
\end{theorem}
\begin{proof}
This is implied by Lemmas \ref{lemma:scgssubsetcs} and \ref{lemma:cssubsetscgs}.
\end{proof}

\begin{theorem}
Any context-sensitive language can be generated by $SCGS$, where each scattered context rule has at most two components.
\end{theorem}
\begin{proofbi}
Obviously, only the first subset of $P_{ABtoCD}^{1}$  has more than two components in its rules. Rules of this subset can be simulated by introduction of some auxiliary rules and symbols. Suppose  rule $(\fplus{X},A,B)\rightarrow (\mcf{X},\rcs{A},\lcs{B})$ and sentential form $\fplus{X}AB$. This rule can be simulated by using those auxiliary rules in a following way:$$\fplus{X}AB\Rightarrow\aux{X}\aorcs{C}B\Rightarrow \aux{X}\rcs{C}\atlcs{D}\Rightarrow \mcf{X}\rcs{C}\lcs{D},$$ where always pairs of symbols are rewritten during each derivation step. Formal proof is left to the reader.
\end{proofbi}


The modified version of $\mathcal{L}(PSCG)=\mathcal{L}(CS)$ problem was discussed in this paper. This modification deals with combination of CD grammar systems with propagating scattered context components and compares their generative power with context-sensitive grammars. The algorithm that constructs grammar system that simulates given context-sensitive grammar has been described. Based on this algorithm, it is shown that those two models have the same generative power. Furthermore it is shown that this property holds even for the most simple variant of these grammar systems---that is, those using only two components, where each scattered context rule is of degree of at most two.

\section*{Acknowledgment}

This work was supported by The Ministry of Education, Youth and Sports of the Czech Republic from the National Programme of Sustainability (NPU II); project IT4Innova\-tions excellence in science - LQ1602; the TA\v{C}R grant TE01020415; and the BUT grant FIT-S-17-3964. 

\nocite{*}
\bibliographystyle{eptcs}
\bibliography{generic}
\end{document}